\documentclass[aps,prx,reprint,twocolumn,
 nofootinbib,showkeys]{revtex4-2}
\usepackage{mathtools,amsthm,amssymb,diffcoeff}
\usepackage{stmaryrd,pifont}
\usepackage{enumitem,multirow,booktabs}
\usepackage{xspace}

\usepackage{hyperref}
\usepackage{cleveref}

\newcommand{\PPT}{PPT\xspace}
\newcommand{\NISQ}{NISQ\xspace}
\newcommand{\iid}{i\@.i\@.d\@.\xspace}
\newcommand{\NHW}{NHW\xspace}
\newcommand{\UTNPW}{UTNPW\xspace}
\newcommand{\NPW}{NPW\xspace}
\newcommand{\EW}{EW\xspace}
\newcommand{\numset}[1]{\mathbb{#1}}
\newcommand{\normaldist}{\ensuremath{\mathcal{N}(0,1)}}

\DeclareMathOperator{\Tr}{Tr}

\theoremstyle{remark}

\newtheorem{Thrm}{Theorem}
\newtheorem{Cor}{Corollary}
\newtheorem{Rem}{Remark}
\crefname{Def}{Definition}{Definitions}
\crefname{Lemma}{Lemma}{Lemmas}
\crefname{Thrm}{Theorem}{Theorems}
\crefname{Cor}{Corollary}{Corollaries}
\crefname{Rem}{Remark}{Remarks}
\crefname{equation}{Eq.}{Eqs.}
\crefname{section}{Sec.}{Secs.}
\crefname{enumi}{Step}{Steps}

\begin{document}
\title{From Certifying Rank $\boldsymbol{k}$ Projectors To Non-Positivity,
 Entanglement, And Non-Hermitian Witnesses Through Traces Of Matrix Powers}

\author{H. F. Chau}
\email[email: ]{hfchau@hku.hk}
\affiliation{Department of Physics, University of Hong Kong, Pokfulam Road,
 Hong Kong}
\date{\today}


\begin{abstract}
 We all know that a density matrix $\rho$ is pure if and only if $\Tr \rho^2 = 1$.  But is this the only way to prove the purity of $\rho$ using the trace of its powers?  Here I systematically study the necessary and sufficient conditions of the more general question of guaranteeing that all eigenvalues of a Hermitian matrix belong to a specified set through its trace of powers.  More importantly, these characterization results are automatically witnesses certifying a quantum state as entangled, a Hermitian operator has least one negative eigenvalue and a linear operator as non-Hermitian.  I demonstrate the effectiveness of these witnesses, analyze their performance, and study their strength, weakness together with resource requirement through numerical simulation as well as analytical work.  I also discuss briefly the effects of numerical stability, uncertainty in measurement and rounding errors on these problems.
\end{abstract}

\maketitle

\section{Introduction}
\label{Sec:Intro}
 It is well-known that a density matrix $\rho$ is a rank one projector (in other words, a pure state) if and only if $\Tr\rho^2 = 1$.  However, this result is no longer true if $\rho$ is replaced by a Hermitian matrix $H$ of dimension greater than two.  In this case, it is known that a necessary and sufficient condition for $H$ to be a rank one projector is $\Tr H = \Tr H^2 = \Tr H^3 = 1$~\cite[{See, for example, Ref.}][]{SchumacherWestmoreland}.  These results mean that the values for a finite set of real numbers (in this case, the real eigenvalues of a matrix) can sometimes be fixed up to permutation through a constant number of constraints independent of the size of the set (here in the form of the first two or three moments of elements in the set).  Here I generalize the above results in multiple ways.  My objectives are to investigate various generalizations and applications of these results.

 I emphasize here that using the trace of matrix powers in certification is not a pure mathematical exercise, at least for Hermitian operators as well as density matrices and their positive partial transposes (\PPT{s}).  It is a useful tool in diverse fields such as certifying entanglement~\cite{entanglement_certification_review1,entanglement_certification_review2,Renyi_entropy_estimation}, entanglement detection~\cite{Renyi_entropy_estimation,entanglement_detection_review1,entanglement_detection_review2}, studying entanglement in condensed matter systems~\cite{Renyi_entropy_estimation,CM_review1,CM_review2}, deducing eigenvalues of a quantum state~\cite{entanglement_spectroscopy1,entanglement_spectroscopy2,SLLJ25} and estimating nonlinear function of quantum states~\cite{SLLJ25}.  Besides, various resource-efficient algorithms to compute or estimate the power of Hermitian matrices or density matrices are available in the literature.  Some of them can be used to show entanglement in existing experiments~\cite{EKetal20}.  Others can be readily implemented using \NISQ devices~\cite{GTC21,SY21,SLLJ25}.  Recently, an experiment on using second moment of the partial transpose to detect entanglement was also announced~\cite{PPT_detection}.

 Is it possible to prove the purity of a density matrix through traces of its higher powers?  Can the above results be extended to rank~$k$ projectors?  Can the non-zero eigenvalues of a Hermitian matrix $H$ be narrow down to, say, two possible values by specifying the traces of just five different powers of $H$?  And how about using only four different powers of $H$?  More generally, what is the most economical way to do so measured in terms of the number of such trace of powers?  These are some of the questions that I am going to answer in \Cref{Sec:trace_powers}.  Precisely, I am going to report a few necessary and sufficient conditions for a complex-valued matrix to be Hermitian as well as for a Hermitian matrix (whose trace may or may not be specified) to be non-negative both with at most $p$ distinct real eigenvalues through specifying a few trace of matrix powers.  Moreover, I prove an explicit expression on the minimum number of such trace of matrix powers needed.  In addition, practical numerical analysis and related approximation issues will also discussed in \Cref{Sec:trace_powers}.  Actually, the closest prior arts to this work are the paper by De las Cuevas \emph{et al.}~\cite{CFN20} who studied to how determine if a Hermitian matrix is positive semidefinite given the trace of its first few moments plus the papers by Elben \emph{et al.}~\cite{EKetal20}, Yu \emph{et al.}~\cite{YIG21} and Neven \emph{et al.}~\cite{Netal21} on forming entanglement witnesses (\EW{s}) through several moments of the \PPT of a density matrix.  I will compare their approaches with mine in \Cref{Sec:trace_powers}.

 Then I turn the results reported in \Cref{Sec:trace_powers} into various different type of witnesses.  They include certifying a complex-valued matrix is not Hermitian as well as guaranteeing a Hermitian matrix to have at least one negative eigenvalue.  Through the famous \PPT criterion, I also obtain a number of \EW{s}.  These witnesses together with their performance will be reported in \Cref{Sec:witness}.  In particular, I find that for uniformly and randomly generated matrices, specifying the values of trace of up to eight different integral powers, the corresponding witnesses perform better than all the published results I aware of.  Actually, my results generalize some of those in Refs.~\cite{EKetal20,YIG21,Netal21}.  In \Cref{Subsubsec:non-positive_witness_performance}, I further discuss the reasons why these witnesses are so effective.  Besides, I report the common features for the failure cases there, too.  Here I remark that showing a matrix is non-Hermitian through its trace of matrix powers is more than a pure theoretical study.  By means of weak measurement or interferometry techniques, expectation value of a non-Hermitian operator can be determined~\cite{trace_non-Hermitian}.  Actually, such experiment has been carried out~\cite{trace_non-Hermitian_expt}.  Finally, conclusions are drawn in \Cref{Sec:conclusion}.

\section{Constraining Eigenvalues Of Hermitian Matrices Through Trace Of Their
 Powers}
\label{Sec:trace_powers}

\subsection{General Hermitian Matrices}
\label{Subsec:trace_powers_general}

\begin{Thrm}
 Let $H$ be a Hermitian matrix and $C_H = \{ c_i \}_{i=1}^p$ be a set of $p$ distinct real numbers.  Then
 \begin{equation}
  \smashoperator{\sum_{j=0}^{\deg g(x) + \sum_i n_i}} \quad a_j \Tr H^j \ge 0
  \label{E:Hermitian_eigens}
 \end{equation}
 where $a_j$ is the coefficient of the $x^j$ term in the polynomial
 \begin{equation}
  G_H(x) \equiv g(x) \prod_{i=1}^p (c_i - x)^{n_i}
  \label{E:G_H_def}
 \end{equation}
 for all $j$.  Here, $g(x)$ is a polynomial obeying $g(x) > 0$ for all $x\in \numset{R}$, and $n_i$ are positive even integers for all $i$.  Furthermore, Inequality~\eqref{E:Hermitian_eigens} becomes an equality if and only if eigenvalues of $H$ are in the set $C_H$.  And in this case, $c_i$ is an eigenvalue of $H$ of multiplicity $k_i$ for all $i$ if and only if $\Tr H^j = \sum_i k_i^{\vphantom{j}} c_i^j$ for all $j = 0,1,2,\dots,p-1$.
 \label{Thrm:Hermitian_eigens}
\end{Thrm}

\begin{proof}
 Clearly, $G_H(x) \ge 0$ with equality holds if and only if $x \in C_H$.
 Let $\lambda_j$'s be the eigenvalues of $H$.  Then the first part of this
 Theorem follows from the observation that $\sum_j G_H(\lambda_j) \ge 0$ with
 equality holds if and only if $\lambda_j \in C_H$ for all $j$.

 As for the second part, note that the Vandermonde matrix $(c_i^j)_{ij}$ is
 invertible.  Hence, there is an one-to-one correspondence between
 $(k_i)_{i=1}^p$ and $(\Tr H^j)_{j=0}^{p-1}$.
\end{proof}

\begin{Rem}
 By putting $n_i = 2$ and $g(x) = 1$, \Cref{Thrm:Hermitian_eigens} implies
 that $(2p+1)$ constraints, namely, the values of $\Tr H^j$ for $j=0,\dots,2p$
 are enough to assure that a Hermitian matrix has $p$ distinct non-zero
 eigenvalues $c_i$'s each with multiplicity $k_i > 0$.  In the event that one
 of the eigenvalues is $0$, it needs $(2p-1)$ constraints through the values
 of $\Tr H^j$ for $j=2,\dots,2p$.  In other words, the resource needed to
 certify the eigenvalues of a highly degenerate Hermitian matrix is low.  This
 finding echos that by Shin \emph{et al.} who discovered that $\Tr \rho^j$ for
 a density matrix $\rho$ and $j$ is a large integer can be well estimated by
 the trace of a few low powers of $\rho$ provided that its rank is
 small~\cite{SLLJ25}.  Actually, these two results above are optimal in the
 sense that one cannot nail down the $p$ distinct non-zero eigenvalues of $H$
 and their multiplicities by specifying the values of $\Tr H^{q_1}, \dots, \Tr
 H^{q_{2p}}$ for $q_j \in \numset{N}$.  The reason is that for the function
 $f(x_1,\dots,x_{2p+1}) = (\sum_{i=1}^{2p+1} x_i^{q_1}, \dots,
 \sum_{i=1}^{2p+1} x_i^{q_{2p}})$, $\{ \diffp{f}{x_i} \}_{i=1}^{2p+1}$ is a
 linearly dependent set at every point $(x_1,\dots,x_{2p+1})$.  So, $f$ is
 not injective and hence using the values of $\Tr H^{q_j}$ for $j=1,\dots,2p$
 are insufficient to determine the set of eigenvalues of $H$ and their
 multiplicities.
 \label{Rem:optimal_number_of_constraints}
\end{Rem}

 The followings are immediate and important consequences of
 \Cref{Thrm:Hermitian_eigens}.

\begin{Cor}
 Regarding $a_j$ in \Cref{Thrm:Hermitian_eigens} as a function of $c_i$'s, a
 Hermitian matrix $H$ has up to $p$ distinct eigenvalues if and only if
 $\displaystyle \inf_{c_1,\dots,c_p \in \numset{R}} \sum_j a_j \Tr H^j = 0$.
 \label{Cor:Hermitian_eigens}
\end{Cor}

\begin{Cor}
 Suppose $H$ is a Hermitian matrix satisfying $\Tr H^\ell = c^\ell k$, $\Tr H^m = c^m k$ and $\Tr H^n = c^n k$ where $\ell,m,n,k \in \mathbb{Z}^+$, $\ell > m > n$ and $0 \ne c \in \numset{R}$.  Then, $H = c \,\Pi$ where $\Pi$ is a rank $k$ orthogonal projector if and only if $\ell,n$ are even and $m$ is odd.
 \label{Cor:rank_k_extension}
\end{Cor}

\begin{proof}
 By substituting $H$ by $H/c$, one only needs to consider the case of $c = 1$.
 By elementary calculus, I know that the function $\displaystyle h(x) = 1 -
 \frac{\ell-n}{\ell-m} \ x^{m-n} + \frac{m-n}{\ell-m} \ x^{\ell-n} \ge 0$ for
 all $x \in \numset{R}$ with equality holds if and only if $x = 1$.  Besides,
 $x = 1$ is its double root.  Hence, by setting $G_H(x) = x^n h(x)$ and $C_H =
 \{ 0, 1 \}$, I conclude that $G_H(x) \ge 0$ with equality holds if and only
 if $x \in C_H$.  So, $\Tr H^j = k$ for all $j = \ell,m,n$ implies that
 Inequality~\eqref{E:Hermitian_eigens} is an equality.  So by
 \Cref{Thrm:Hermitian_eigens}, $H$ must be a rank $k$ orthogonal projector.
 Conversely, if $H$ is a rank $k$ projector, surely $\Tr H^j = k$ for all $j
 \in \numset{Z}^+$.  This completes the proof.
\end{proof}

\begin{Rem}
 \Cref{Cor:Hermitian_eigens} can be viewed as finding all the distinct
 eigenvalues of a Hermitian matrix $H$ by varying over the set $C_H$ so that
 the value of the LHS of Inequality~\eqref{Thrm:Hermitian_eigens} becomes $0$.
 Similar strategy works for finding eigenvalues of a non-negative Hermitian
 matrix through
 \Cref{Thrm:positive_eigens_for_unit_trace_H,Cor:positive_eigens} to be
 reported in \Cref{Subsec:trace_powers_non-negative} below.
 \label{Rem:determining_C_H}
\end{Rem}

\begin{Rem}
 \Cref{Cor:rank_k_extension} demonstrates that by properly chosen $g(x)$ and
 $n_i$'s in \Cref{Thrm:Hermitian_eigens}, one may greatly reduce the number of
 non-zero coefficients $a_j$ even though the degree of $G_H$ may be high.
 Physically, it means that the knowing the traces of a few possibly high
 powers of $H$, whose numbers can be as small as the optimal value given by
 \Cref{Rem:optimal_number_of_constraints} is enough to restrict its
 eigenvalues to a small set $C_H$.
\end{Rem}

 The two Theorems below show that using polynomial $G_H$ is not the only way
 to obtain necessary and sufficient conditions on eigenvalues of $H$ through
 specifying some of its trace powers.  In particular,
 \Cref{Thrm:rank_one_projector_iff2} generalizes the known result that a
 Hermitian matrix $H$ is a rank one projector if and only if $\Tr H = \Tr H^2
 = \Tr H^3 = 1$~\cite{SchumacherWestmoreland}.

\begin{Thrm}
 Suppose $H$ is a Hermitian matrix, $m$ is an odd positive integer and $n$ is an even positive integer.  Then,
 \begin{equation}
  (\Tr H^n)^2 \ge \Tr H^{2n} .
  \label{E:rank_one_projector_iff1_inequality}
 \end{equation}
 Moreover, suppose $c \in \numset{R}$ is non-zero.  Then, $H/c$ is a rank one
 projector if and only if $\Tr H^m = c^m$, $\Tr H^n = c^n$, $\Tr H^{2n} =
 c^{2n}$.
 \label{Thrm:rank_one_projector_iff1}
\end{Thrm}

\begin{proof}
 Observe that $(\Tr H^n)^2 - \Tr H^{2n} = \sum_{i\ne j} \lambda_i^n
 \lambda_j^n \ge 0$, where $\lambda_i$'s are the eigenvalues of $H$.  This
 proves Inequality~\eqref{E:rank_one_projector_iff1_inequality}.

 To continue, by considering $H/c$, I may assume that $c = 1$.  Then $\Tr H^n
 = 1 = \Tr H^{2n}$ implies $\Tr H^{2n} = (\Tr H^n)^2$, which in turn implies
 $\sum_{i\ne j} \lambda_i^n \lambda_j^n = 0$.  As $n$ is even, this happens if
 and only if at most one of the eigenvalues, say $\lambda_1$, is non-zero.
 Finally, $\Tr H^m = 1$ with $m$ being odd forces $\lambda_1 = 1$.  In other
 words, $H$ is a rank one projector.  The converse is trivially true.
\end{proof}

\begin{Thrm}
 Let $H$ be a Hermitian matrix, $m$ is an odd positive integer.  Then
 \begin{equation}
  (\Tr H^{2m})^2 - 2 \Tr H^{3m} + \Tr H^{2m} \ge 0 .
  \label{E:rank_one_projector_inequality}
 \end{equation}
 Suppose $0\ne c \in \numset{R}$.  Then, $H/c$ is a rank one projector if and
 only if $\Tr H^m = c^m$, $\Tr H^{2m} = c^{2m}$, $\Tr H^{3m} = c^{3m}$.
 \label{Thrm:rank_one_projector_iff2}
\end{Thrm}

\begin{proof}
 The Inequality~\eqref{E:rank_one_projector_inequality} holds because its RHS
 equals $\sum_i \lambda_i^{2m} (\lambda_i^m - 1)^2 + \sum_{i\ne j}
 \lambda_i^{2m} \lambda_j^{2m} \ge 0$, where $\lambda_i$'s are the eigenvalues
 of $H$.  Besides, the equality holds if and only if $\lambda_i \in \{ 0,1 \}$
 for all $i$ and at most one of the eigenvalues is non-zero.

 Since $m$ is odd and $H$ is Hermitian, $H^{1/m}$ is well-defined and
 Hermitian.  By considering $(H/c)^{1/m}$, I may assume that $c = m = 1$ from
 now on.  In this case, from the derivation of the previous paragraph, $\Tr H
 = \Tr H^2 = \Tr H^3 = 1$ together with $\Tr H = 1$ implies that exactly one
 of the eigenvalues of $H$ is $1$ while the rest are all $0$.  Hence, $H$ is a
 rank one projector.  The proof of the converse is straightforward.
\end{proof}

\subsection{Non-Negative Hermitian Matrices}
\label{Subsec:trace_powers_non-negative}

 Some proof techniques in \Cref{Subsec:trace_powers_general} can be readily
 applied to the case of eigenvalues determinations of non-negative Hermitian
 matrices.

\begin{Thrm}
 Let $C_{NN} = \{ 0, 1, c_1,\dots, c_p \}$ where $c_i \in (0,1)$ are distinct.  Suppose $H$ is a unit trace non-negative Hermitian matrix.  Then,
 \begin{equation}
  \smashoperator{\sum_{j=0}^{\substack{\deg g(x) + m_1 + \\ m_2 + \sum_i n_i}}}
  \enspace b_j \Tr H^j \ge 0
  \label{E:non-negative_Hermitian_eigens_inequality}
 \end{equation}
 where $b_j$ is the coefficient of $x^j$ in the polynomial
 \begin{equation}
  G_{NN}(x) = g(x) x^{m_1} (1-x)^{m_2} \prod_{i=1}^p (c_i - x)^{n_i}
  \label{E:G_NN_def}
 \end{equation}
 for all $j$ with $g(x)$ being a polynomial obeying $g(x) > 0$ for all $x\in [0,1]$, $m_1, m_2 \in \numset{N}$, and $n_i$ are even positive integers for all $i$.  Moreover, if $m_1, m_2 > 0$, then eigenvalues of $H$ are in $C_{NN}$ if and only if Inequality~\eqref{E:non-negative_Hermitian_eigens_inequality} is an equality.  And in this case, $c_i$ is an eigenvalue of $H$ of multiplicity $k_i$ for all $i$.  Besides $1$ is an eigenvalue of $H$ of multiplicity $k$ if and only if $\Tr H^j = k + \sum_i k_i c_i^j$ for $j = 0,1,2,\dots,p$.
 \label{Thrm:positive_eigens_for_unit_trace_H}
\end{Thrm}

\begin{proof}
 Note that $G_{NN}(x) \ge 0$ for all $x\in [0,1]$.  In addition, if $m_1, m_2
 > 0$, then $G_{NN}(x) = 0$ if and only if $x \in C_{NN}$.  Since $H$ is
 Hermitian with unit trace, all its eigenvalues must lie in $[0,1]$ if $H \ge
 0$.  The same proof technique as that of \Cref{Thrm:Hermitian_eigens} can now
 be used to show the validity of this Theorem.
\end{proof}

\begin{Rem}
 By setting $G_{NN}(x) = x(1-x)$, \Cref{Thrm:positive_eigens_for_unit_trace_H}
 reduces to the well-known result that a density matrix $\rho$ is pure if and
 only if $\Tr \rho^2 = 1$~\cite{SchumacherWestmoreland}.  More generally, by
 using $G_{NN}(x) = x^{q_1} (1-x^{q_2})$ with $q_1, q_2 \in \numset{Z}^+$,
 \Cref{Thrm:positive_eigens_for_unit_trace_H} implies that $\rho$ is pure if
 and only if $\Tr \rho^{q_1} = \Tr \rho^{q_1 + q_2} = 1$. 
 \label{Rem:determining_C_NN}
\end{Rem}

\begin{Cor}
 Suppose one is given the same conditions and using the same notations as in
 \Cref{Thrm:positive_eigens_for_unit_trace_H} except that $\Tr H$ is not
 known and that $C_{NN}$ is replaced by $C'_{NN} = \{0, c_1, \dots, c_p \}$
 with $c_i > 0$ for all $i$.  Then, $H \ge 0$ implies
 \begin{equation}
  \smashoperator{\sum_{j=0}^{\substack{\deg g(x) + \\ m_1 + \sum_i n_i}}}
  \enspace b'_j \Tr H^j \ge 0
  \label{E:non-negative_Hermitian_eigens_alt}
 \end{equation}
 where $b'_j$ is the coefficient of $x^j$ in the polynomial
 \begin{equation}
  G'_{NN}(x) = g(x) x^{m_1} \prod_{i=1}^p (c_i - x)^{n_i}
  \label{E:G_NN_alt_def}
 \end{equation}
 for all $j$.  Besides, Inequality~\eqref{E:non-negative_Hermitian_eigens_alt} is an equality if and only if eigenvalues of $H$ are in $C'_{NN}$.  And in this case, $c_i$ is an eigenvalue of $H$ of multiplicity $k_i$ for all $i$ if and only if $\Tr H^j = \sum_i k_i c_i^j$ for $j = 0,1,2,\dots,p-1$.
 \label{Cor:positive_eigens}
\end{Cor}

\begin{proof}
 Following the logic in the proof of
 \Cref{Thrm:positive_eigens_for_unit_trace_H}, as $\Tr H$ is not known, I need
 to construct a polynomial in $x$ that is non-negative for all $x\ge 0$.
 Besides, elements of $C'_{NN}$ are roots of this polynomial.  Surely, one
 such polynomial is $G'_{NN}$.  The rest of the proof is the same as that of
 \Cref{Thrm:positive_eigens_for_unit_trace_H}.
\end{proof}

\begin{Cor}
 Let $H$ be a non-negative Hermitian matrix obeying $\Tr H^m = c^m k$, $\Tr
 H^n = c^n k$ and $\Tr H^t = c^t k$ where $m,n,t,k \in \numset{Z}^+$, $m > n >
 t$ and $c > 0$.  Then, $H = c \,\Pi$ if and only if $m, t$ are odd and $n$ is
 even.
 \label{Cor:positive_eigens_rank_k_extension}
\end{Cor}

\begin{proof}
 Using the idea in the proof of \Cref{Cor:rank_k_extension}, let
 $\displaystyle G'_{NN}(x) = x^t \left( 1 - \frac{m-t}{m-n} \,x^{n-t} +
 \frac{n-t}{m-n} \,x^{m-t} \right)$.  As $G'_{NN}(x) \ge 0$ for all $x \ge 0$
 with equality holds if and only $x \in \{ 0, c \}$, this Corollary follows
 from \Cref{Cor:rank_k_extension,Cor:positive_eigens}.
\end{proof}

 The following Theorem is the common counterpart of
 \Cref{Thrm:rank_one_projector_iff1,Thrm:rank_one_projector_iff2} for
 non-negative Hermitian matrices.

\begin{Thrm}
 Suppose $H$ is a non-negative Hermitian matrix, $m$ is an odd positive integer, and $c \ge 0$.  Then,
 \begin{equation}
  (\Tr H^m)^2 \ge \Tr H^{2m} .
  \label{E:non-negative_rank_one_projector_inequality}
 \end{equation}
 Moreover, $H = c \,\Pi$ where $\Pi$ is a rank one projector if and only if $\Tr H^m = c^m$ and $\Tr H^{2m} = c^{2m}$.
 \label{Thrm:non-negative_rank_one_projector_iff}
\end{Thrm}

\begin{proof}
 Denote the eigenvalues of $H$ by $\lambda_i$'s.  Since $H \ge 0$, I know that
 $(\Tr H^m)^2 - \Tr H^{2m} = \sum_{i\ne j} \lambda_i^m \lambda_j^m \ge 0$ with
 equality holds if and only if at most one of the eigenvalues, say,
 $\lambda_1$ is positive.  Since $m$ is odd, the value of $\lambda_1$ equals
 $(\Tr H^m)^{1/m}$.  This completes the proof.
\end{proof}

 A common strategy used in the proofs of all the Theorems and Corollaries in
 this paper is to fix the values of the eigenvalues of $H$ up to permutation
 through constraints on the trace of its powers.  It pays no attention to the
 diagaonalizability of $H$.  Hence, these Theorems and Corollaries are valid
 if the Hermitian or non-negative $H$ is replaced by a complex-valued matrix
 with real or non-negative eigenvalues respectively provided that statements
 such as rank~$k$ orthogonal projectors are properly substituted, in this
 example, by exactly $k$ non-zero eigenvalues (measured in algebraic
 multiplicity) and that they all take the value of $1$.

 The methods used in showing
 \Cref{Thrm:rank_one_projector_iff1,Thrm:rank_one_projector_iff2,%
 Thrm:non-negative_rank_one_projector_iff} can be recast in the language of
 algebraic variety.  Nonetheless, I do not pursue this type of construction
 further because of two reasons.  First, there is no general construction
 strategy.  More importantly, one will see in
 \Cref{Subsec:witness_performance} that the polynomial construction method
 used in \Cref{Thrm:Hermitian_eigens} and the like plus
 \Cref{Thrm:rank_one_projector_iff1,Thrm:rank_one_projector_iff2,%
 Thrm:non-negative_rank_one_projector_iff} are already sufficient to give
 extremely powerful witnesses of various kinds.

 Now, I compare the strategy used here with those in
 Refs.~\cite{CFN20,EKetal20,YIG21,Netal21}.  By choosing $G_{NN}(x) =
 x(c-x)^2$ and minimizing the LHS of
 Inequality~\eqref{E:non-negative_Hermitian_eigens_inequality} over $c\in
 \numset{R}$ just like in \Cref{Rem:determining_C_H}, one gets
\begin{equation}
 \Tr H \Tr H^3 - (\Tr H^2)^2 \ge 0
 \label{E:TrH123}
\end{equation}
 for all $H \ge 0$.  This result and the proof method are
 essentially identical to those in Refs.~\cite{EKetal20,Netal21}.  In
 other words, \Cref{Thrm:positive_eigens_for_unit_trace_H} generalizes their
 results.

 In Ref.~\cite{CFN20}, De las Cuevas \emph{et al.} first showed that in
 Schatten $p$-norm the distance between a Hermitian matrix $H$ from the
 positive semi-definite cone is a function of $\dim H$ and the sum of the
 $p$th power of the absolute value of the negative eigenvalues of $H$.  They
 then used three different methods of decreasing accuracy and computational
 complexity to bound this distance, resulting in three different sufficient
 conditions for $H$ to be non-negative~\cite{CFN20}.  In contrast, methods
 reported in this paper treat positive and negative eigenvalues equally in the
 proof.  Besides, results here need not depend on $\dim H$ so long as $b_0$ or
 $b'_0 = 0$.

 Yu \emph{et al.} gave an alternative derivation of
 Inequality~\eqref{E:TrH123} plus its extensions.  Specifically, they proved
 a family of necessary conditions for $H \ge 0$ through certain Hankel
 matrices whose elements are trace of powers of $H$~\cite{YIG21}.  The first
 member of this family is Inequality~\eqref{E:TrH123}.  Their approach is very
 different from the one used here.

\subsection{Practical Numerical Error Analysis And Matrices With Eigenvalues
 Close To A Given Set}
\label{Subsec:trace_powers_numerics}
 In practice, one usually needs to show that the eigenvalues of a matrix is in
 the set $C_H, C_{NN}$ or $C'_{NN}$ where the number of elements in these sets
 is small.  Besides, one uses only a handful of trace powers in all the
 Theorems and Corollaries reported in
 \Cref{Subsec:trace_powers_general,Subsec:trace_powers_non-negative}.  In this
 situation, numerical errors in computing the values of $a_j$'s, $b_j$'s and
 $b'_j$'s are ignorable.  Then computing the LHS of
 Inequality~\eqref{Thrm:Hermitian_eigens} via recursive summation is
 numerically stable with the magnitude of the (forward) error $E$ satisfying
 \begin{equation}
  E \lesssim E_\text{bound} \equiv \sum_j |a_j| (2u|\Tr H^j| + |\delta_j|) ,
  \label{E:forward_error_bound}
 \end{equation}
 where $u$ is the machine epsilon of the computer hardware used and $\delta_j$
 is the uncertainty in $\Tr H^j$'s~\cite{NA1,NA2}.  Here I emphasize that the
 relative error can be high due to rounding for $a_j$'s and $b_j$'s alternate
 in sign.  Furthermore, if the trace powers are experimentally determined,
 then their uncertainties $\delta_j$'s are almost always much greater than $u$
 in magnitude at least when $j$ is small.  Consequently, one can safely drop
 the term involving $u$ in \Cref{E:forward_error_bound} when used in real
 experiments.

 In all cases, if the magnitude of the LHS of
 Inequality~\eqref{E:Hermitian_eigens} is much less than $E_\text{bound}$,
 then one could conclude that the eigenvalues of $H$ are consistent with the
 hypothesis that they are all in $C_H$.  At least, they should be close to
 $C_H$ in the sense that, say, $\displaystyle \max_j d(\lambda_j,C_H)$ is
 small for all $j$ where $\lambda_j$'s are eigenvalues of $H$ and
 $\displaystyle d(\lambda_j,C_H) \equiv \min_{i=1}^p |\lambda_j - c_i|$ is the
 Euclidean distance of the point $\lambda_j$ from the set $C_H$.

 Consider the variation below.  Suppose the LHS of
 Inequality~\eqref{E:Hermitian_eigens} equals a fixed number $\epsilon \ge 0$.
 Suppose further that $g(x)$ is not a rapidly varying function in the
 neighborhood of $C_H$.  Without lost of generality, by rearranging the
 indices, I may assume that $n_1 \le n_i$ for all $i$.  Then, by elementary
 calculus, it is straightforward to see that $\sum_j d(\lambda_j,C_H)$ is
 maximized when all $\lambda_j$'s satisfy $|\lambda_j - c_1| =
 \delta_\text{sum}$ where
\begin{subequations}
 \label{E:backward_error_bound}
 \begin{equation}
  \delta_\text{sum} \lesssim \left[ \frac{\epsilon + E_\text{bound}}{g(c_1)
  \dim H \,\prod_{i=2}^p (c_i - c_1)^{n_i}} \right]^{1/n_1} .
 \end{equation}
 Likewise, $\displaystyle \max_j d(\lambda_j,C_H)$ is maximized when exactly
 one $\lambda_j$ is not in $C_H$.  Besides, this $\lambda_j$ obeys
 $|\lambda_j - c_1| = \delta_{\max}$ with
 \begin{equation}
  \delta_{\max} \lesssim \left[ \frac{\epsilon + E_\text{bound}}{g(c_1)
  \prod_{i=2}^p (c_i - c_1)^{n_i}} \right]^{1/n_1} .
 \end{equation}
\end{subequations}

 Similar bounds can be drawn for all other Theorems and Corollaries in
 \Cref{Subsec:trace_powers_general,Subsec:trace_powers_non-negative}.  I omit
 them here to save space as details can be easily filled in by interested
 readers.

 Last but not least, I discuss the numerical stability in finding the
 multiplicities of eigenvalues of $H$.  Recall from the proof of
 \Cref{Thrm:Hermitian_eigens} that one may find these multiplicities $k_i$'s by
 inverting the Vandermonde matrix $(c_i^j)_{ij}$.  (More precisely, one uses
 the closest non-negative integral values of provisional values of $k_i$'s
 upon solving the Vandermonde matrix equation.)  Nevertheless, in general the
 condition number of this matrix grows exponential in $p$, namely, the number
 of $c_i$'s in the set $C_H, C_{NN}$ or $C'_{NN}$~\cite{Pan}.  Hence, finding
 the multiplicities of $k$ and $k_i$'s through values of $\Tr H^j$'s using the
 system of linear equations stated in \Cref{Thrm:Hermitian_eigens,%
 Thrm:positive_eigens_for_unit_trace_H} as well as \Cref{Cor:positive_eigens}
 is numerically unstable when $p$ is large even if all the $c_i$'s are
 precisely known.

\section{Various Witnesses}
\label{Sec:witness}
 Each violation of a necessary condition for a matrix to be Hermitian (non-negative) can be used as a non-Hermitian (non-positive\footnote{Here I use the term non-positive instead of the technically more correct but unnatural terms of ``non-non-negative'' or ``non-semipositive definite'' to describe a Hermitian matrix with at least one negative eigenvalue.}) witness.  Surely, combined with the famous \PPT criterion for the separability of a density matrix~\cite{PPT}, every non-positive witness is also an \EW.  I write down these witnesses below.

\subsection{Formal Statements}
\label{Subsec:witness_statements}

\begin{Thrm}[Non-Hermitian Witness {[}\NHW{]}]
 Let $a_j$ be the coefficient of $x^j$ defined in \Cref{Thrm:Hermitian_eigens}, which in turn is a function of $c_i$'s.  Then a matrix $M$ is not Hermitian if
 \begin{widetext}
  \begin{subequations}
   \label{E:non-Hermitian_witness}
   \begin{gather}
    \Im (\Tr M^q) \ne 0 , \text{~or}
    \label{E:imag_part_non-Hermitian_witness} \\
    \left[ \inf_{c_1,\dots,c_p \in \numset{R}}
     \smashoperator{\sum_{j=0}^{\deg g(x) + \sum_i n_i}} \quad
     a_j(c_1,\dots,c_p) \,\Re (\Tr M^j) < 0 , \enspace \Im (\Tr M) = \Im (\Tr
     M^2) = \dots = 0 \right] , \text{~or}
     \label{E:optimizing_non-Hermitian_witness} \\
    \left\{ [\Re (\Tr M^n)]^2 < \Re (\Tr M^{2n}), \enspace \Im (\Tr M^n) = \Im
     (\Tr M^{2n}) = 0 \right\} , \text{~or}
     \label{E:12_non-Hermitian_witness} \\
    \left\llbracket \left\{ [\Re (\Tr M^{2m})]^3 < [\Re (\Tr M^{3m})]^2 ,
     \enspace \Re (\Tr M^{2m}) > 0 \} \text{~or~} \Re (\Tr M^{2m}) < 0
     \right\}, \enspace \Im (\Tr M^{2m}) = \Im (\Tr M^{3m}) = 0
     \right\rrbracket
     \label{E:123_optimized_non-Hermitian_witness}
   \end{gather}
  \end{subequations}
 \end{widetext}
 for $q \in \numset{Z}^+$ and some even (odd) positive integer $n$ ($m$).  Notice that any one of the above inequalities is also a non-Hermitian witness (\NHW) if $M$ is a linear operator provided that the values of $\Re (\Tr M^j)$'s and $\Im (\Tr M^j)$'s involved are well-defined.
 \label{Thrm:non-Hermitian_witness}
\end{Thrm}

\begin{proof}
 The only non-trivial witness that requires discussion is
 Inequality~\eqref{E:123_optimized_non-Hermitian_witness}.
 \Cref{Thrm:rank_one_projector_iff2} implies that $\Re [\Tr (M/c)^{2m}]\}^2 -
 2 \Re [\Tr (M/c)^{3m}] + \Re [\Tr (M/c)^{2m}] < 0$ and $\Im (\Tr M^{2m}) =
 \Im (\Tr M^{3m}) = 0$ guarantee $M$ to be non-Hermitian.
 Inequality~\eqref{E:123_optimized_non-Hermitian_witness} is then obtained by
 minimizing $c$.

 In the case that $M$ is a linear operator, the assertion follows from the
 following observation.  If the relevant $\Re (\Tr M^j)$'s are well-defined,
 then using the same argument in the proof of \Cref{Thrm:Hermitian_eigens},
 I know that LHS of Inequality~\eqref{E:Hermitian_eigens} is non-negative
 provided that $M$ is Hermitian.  Consequently,
 Inequality~\eqref{E:optimizing_non-Hermitian_witness} is a \NHW.  All
 remaining inequalities in this Theorem can be shown to be valid \NHW{s} in a
 similar way.
\end{proof}

\begin{Rem}
 Several \NHW{s} in the form
 Inequality~\eqref{E:optimizing_non-Hermitian_witness} can be simplified.  For
 example, putting $G_H(x) = x^n$ for some even positive integer $n$ gives the
 \NHW
 \begin{subequations}
  \label{E:optimizing_non-Hermitian_witness_special_cases}
  \begin{equation}
   \Re(\Tr M^n) < 0 .
  \label{E:2_optimizing_non-Hermitian_witness_special_cases}
  \end{equation}
  By setting $G_H(x) = (c^q - x^q)^2$ for $q \in \numset{Z}^+$, the
  corresponding \NHW is
  \begin{equation}
   \dim M \,\Re (\Tr M^{2q}) < [\Re (\Tr M^q)]^2 .
   \label{E:012_optimizing_non-Hermitian_witness_special_cases}
  \end{equation}
  And by using the $G_H(x)$ in the proof of \Cref{Cor:rank_k_extension}, the
  corresponding \NHW is
  \begin{align}
   & \Re (\Tr M^n) c^{\ell-n} - \frac{\ell - n}{\ell - m} \,\Re (\Tr M^m)
    c^{\ell - m} + \notag \\
   & \frac{m-n}{\ell - m} \,\Re (\Tr M^\ell) < 0 \enspace \text{and~}
    \Re (\Tr M^n) > 0
   \label{E:lmn_optimizing_non-Hermitian_witness_special_cases}
  \end{align}
 \end{subequations}
 with $c = [\Re (\Tr M^m)/\Re (\Tr M^n)]^{1/(m-n)}$ for $\ell > m > n > 0$
 with $\ell, n$ being even integers and $m$ being an odd integer,
 respectively.  (Since $m-n$ is an odd positive integer, $c$ is a
 well-defined real number.)
 \label{Rem:non-Hermitian_witness}
\end{Rem}

\begin{Thrm}[Non-Positive Witness For Unit Trace Hermitian Matrices {[}\UTNPW{]}]
 Let $b_j$ be the coefficient of $x^j$ defined in \Cref{Thrm:positive_eigens_for_unit_trace_H}, which in turn is a function of $c_i$'s.  Then a unit trace Hermitian matrix $H$ is non-positive in the sense that it has at least one negative eigenvalue if
 \begin{subequations}
  \label{E:non-positive_witness_for_unit-trace_Hermitian}
  \begin{gather}
   \min_{c_1,\dots,c_p \in [0,1]}
   \smashoperator{\sum_{j=0}^{\substack{\deg g(x) + m_1 + \\ m_2 + \sum_i n_i}}}
   \enspace b_j(c_1,\dots,c_p) \Tr H^j < 0 , \text{~or}
   \label{E:optimizing_non-positive_witness_for_unit-trace_Hermitian} \\
   (\Tr H^m)^2 < \Tr H^{2m}
   \label{E:12_non-positive_witness_for_unit-trace_Hermitian}
  \end{gather}
 \end{subequations}
 for some odd positive integer $m$.  Besides, at least one of $m_1$ or $m_2$ has to be an odd positive integer.  Moreover, the above inequalities are valid non-positive witnesses for unit trace Hermitian operator (\UTNPW{s}) if $H$ is a unit trace Hermitian operator provided that the values of $\Tr H^j$'s involved are well-defined.
 \label{Thrm:non-positive_witness_for_unit-trace_Hermitian}
\end{Thrm}

\begin{proof}
 This Theorem follows from \Cref{Thrm:positive_eigens_for_unit_trace_H} after
 taking the following two subtleties into account.  First, just like
 Inequality~\eqref{E:optimizing_non-Hermitian_witness} in
 \Cref{Thrm:non-Hermitian_witness}, this witness is optimized through the
 infimum over $c_i$'s in the range $(0,1)$.  Note that from
 \Cref{Thrm:positive_eigens_for_unit_trace_H}, $b_j$ can be regarded as a
 multivariate polynomial of $c_i$'s for all $j$.  Besides, each $c_i \in
 (0,1)$.  Thus, the LHS of
 Inequality~\eqref{E:optimizing_non-positive_witness_for_unit-trace_Hermitian}
 is continuous functions in $c_i$'s.  Therefore, one may replace the infimum by minimum over the closed interval $[0,1]$.  Second, if both $m_1$ and $m_2$
 are even positive integers, then $G_{NN}(x) \ge 0$ for all $x \in
 \numset{R}$.  Thus, the witness in the form of
 Inequality~\eqref{E:optimizing_non-positive_witness_for_unit-trace_Hermitian}
 can never true.  Consequently, the case of $m_1$ and $m_2$ being even can be
 dropped from the witness test.
\end{proof}

 Obviously, same argument in the above proof applies to
 Inequalities~\eqref{E:optimizing_non-positive_witness_for_general_Hermitian}
 in \Cref{Thrm:non-positive_witness_for_general_Hermitian}
 and Inequality~\eqref{E:optimizing_entanglement_witness} in
 \Cref{Thrm:entanglement_witness} below.

\begin{Thrm}[Non-Positive Witness For General Hermitian Matrices {[}\NPW{]}]
 Let $b'_j$ be the coefficient of $x^j$ defined in \Cref{Cor:positive_eigens}, which in turn is a function of $c_i$'s.  Then a Hermitian matrix $H$ is non-positive if
 \begin{subequations}
  \label{E:non-positive_witness_for_general_Hermitian}
  \begin{gather}
   \inf_{c_1,\dots,c_p \ge 0}
   \smashoperator{\sum_{j=0}^{\substack{\deg g(x) + \\ m_1 + \sum_i n_i}}}
   \enspace b'_j(c_1,\dots,c_p) \Tr H^j < 0 , \text{~or}
   \label{E:optimizing_non-positive_witness_for_general_Hermitian} \\
   (\Tr H^m)^2 < \Tr H^{2m}
   \label{E:12_optimizing_non-positive_witness_for_general_Hermitian}
  \end{gather}
 \end{subequations}
 for some odd integers $m_1, m > 0$.  Besides, the above inequalities are valid non-positive witnesses (\NPW{s}) for Hermitian operator if $H$ is a Hermitian operator provided that the values of $\Tr H^j$'s involved are well-defined.
 \label{Thrm:non-positive_witness_for_general_Hermitian}
\end{Thrm}

\begin{Rem}
 Just like \Cref{Rem:non-Hermitian_witness}, a few \NPW{s} from
 Inequality~\eqref{E:optimizing_non-positive_witness_for_general_Hermitian}
 can be simplified.  They include
 \begin{subequations}
  \label{E:optimizing_non-positive_witness_for_general_Hermitian_special_cases}
  \begin{equation}
   \Tr H^m < 0 ,
   \label{E:m_optimizing_non-positive_witness_for_general_Hermitian_special_cases}
  \end{equation}
  \begin{equation}
   (\Tr H) (\Tr H^{2q+1}) < (\Tr H^{q+1})^2 \enspace \text{and~} \Tr H > 0 ,
   \label{E:1q2q_optimizing_non-positive_witness_for_general_Hermitian_special_cases}
  \end{equation}
  for $q \in \numset{Z}^+$, and
  \begin{align}
   & (\Tr H^t) c^{m-t} - \frac{m-t}{m-n} (\Tr H^n) c^{m-n} + \notag \\
   & \frac{n-t}{m-n} (\Tr H^m) < 0 \enspace \text{and~} (\Tr H^t) > 0
   \label{E:mnt_optimizing_non-positive_witness_for_general_Hermitian_special_cases}
  \end{align}
 \end{subequations}
 with $c = (\Tr H^n / \Tr H^t)^{1/(n-t)}$ for $m > n > t > 0$ with $m,t$ being
 odd integers and $n$ being an even integer, respectively.  (Since $n-t$ is a
 positive odd integer, $c$ is a well-defined real number.)
 \label{Rem:non-positive_witness_for_general_Hermitian}
\end{Rem}

\begin{Thrm}[Entanglement Witness {[}\EW{]}]
 Using the notations in \Cref{Thrm:non-positive_witness_for_unit-trace_Hermitian}, a bipartite density matrix $\rho$ a finite--dimensional Hilbert space is entangled if its \PPT $\sigma \equiv \rho^\text{\PPT}$ obeys
 \begin{subequations}
  \label{E:entanglement_witness}
  \begin{gather}
   \min_{c_1,\dots,c_p \in [0,1]}
   \smashoperator{\sum_{j=0}^{\substack{\deg g(x) + m_1 + \\ m_2 + \sum_i n_i}}}
   \enspace b_j(c_1,\dots,c_p) \Tr \sigma^j < 0 , \text{~or}
   \label{E:optimizing_entanglement_witness} \\
   (\Tr \sigma^m)^2 < \Tr \sigma^{2m} ,
  \end{gather}
 \end{subequations}
 for some odd positive integer $m$.  In addition, at least one of $m_1$ or $m_2$ has to be an odd positive integer.  Furthermore, if the underlying Hilbert space of $\rho$ is infinite-dimensional, the same conclusion holds provided that the corresponding $\Tr \sigma^j$'s are well-defined.
 \label{Thrm:entanglement_witness}
\end{Thrm}

\begin{Rem}
 Notice that $a_j, b_j$ and $b'_j$ are multivariate polynomials of $c_j$'s.
 Hence, even finding approximate global infima or minima in
 \Cref{E:optimizing_non-Hermitian_witness,%
 E:optimizing_non-positive_witness_for_unit-trace_Hermitian,%
 E:optimizing_non-positive_witness_for_general_Hermitian,%
 E:optimizing_entanglement_witness} are
 NP-hard in general~\cite{poly_problem1,poly_problem2,poly_problem3,%
 poly_problem4,poly_problem5}.
 Fortunately, fast algorithms exist for generic
 situations~\cite{poly_problem4,poly_problem5,sos-review}.  Besides, there is
 no need to find these solutions in practice most of the time for a witness
 can be established so long as the corresponding inequality hold for a
 particular choice of $c_i$'s.  This greatly speed up computations in almost
 all cases of interest.  Moreover, for the case of $p = 1$ so that $c_1$ is
 the only parameter to optimize, one may exactly find the global infimum.  For
 instance, if one further fix $g(x) = 1$, $m_1 = m_2 = 1$, and $n_1 = 2$, then
 it is straightforward to check that
 Inequality~\eqref{E:optimizing_non-Hermitian_witness} reduces to
 \begin{subequations}
  \begin{equation}
   \dim H \, \Tr H^2 < (\Tr H)^2 ,
   \label{E:reduced_optimizing_non-Hermitian_witness}
  \end{equation}
  Inequality~\eqref{E:optimizing_non-positive_witness_for_unit-trace_Hermitian}
  reduces to
  \begin{gather}
   \Tr H < \Tr H^2 \text{~or~} \notag \\
   \begin{aligned}
    & (\Tr H - \Tr H^2) (\Tr H^3 - \Tr H^4) \\
    <{} & [ \Tr H^2 - 2 \Tr H^3 + \Tr H^4 ]^2 ,
   \end{aligned}
   \label{E:reduced_optimizing_non-positive_witness_for_unit-trace_Hermitian}
  \end{gather}
  and Inequality~\eqref{E:optimizing_entanglement_witness} reduces to
  \begin{equation}
   (\Tr \sigma^3) < (\Tr \sigma^2)^2 ,
   \label{E:reduced_entanglement_witness}
  \end{equation}
 \end{subequations}
 respectively.  And as pointed out at the end of
 \Cref{Subsec:trace_powers_non-negative},
 Inequality~\eqref{E:reduced_entanglement_witness} had been reported in the
 literature~\cite{EKetal20,YIG21,Netal21}.
\end{Rem}

\begin{Rem}
 If
 Inequality~\eqref{E:optimizing_non-positive_witness_for_general_Hermitian}
 derived from the function $G'_{NN}(x)$ in \Cref{E:G_NN_def} is a \NPW for a
 matrix $H$, then so is the inequality corresponding to the function $(c-x)^2
 G'_{NN}(x)$ where $\displaystyle c = \max_{i,j} (c_i,\lambda_j)$ with
 $\lambda_j$ being the eigenvalues of $H$.  This is because $c \ge 0$ and
 hence $(c-x)^2$ is a decreasing function for $x \le c$.  Besides,
 $G'_{NN}(\lambda_j) < 0$ if and only if $\lambda_j < 0$.  As a result,
 \begin{align}
  & \sum_j (c-\lambda_j)^2 \,G'_{NN}(\lambda_j) \notag \\
  ={}& \sum_{j\colon \lambda_j \ge 0} (c-\lambda_j)^2 \,G'_{NN}(\lambda_j) +
   \sum_{j\colon \lambda_j < 0} (c-\lambda_j)^2 \,G'_{NN}(\lambda_j) \notag \\
  \le{}& \sum_{j\colon \lambda_j \ge 0} c^2 \,G'_{NN}(\lambda_j) +
   \sum_{j\colon \lambda_j < 0} c^2 \,G'_{NN}(\lambda_j) \notag \\
  ={}& c^2 \sum_j G'_{NN}(\lambda_j) < 0 .
 \end{align}
 In other words, by increasing the value of $p$ (which means by using more
 $c_i$'s), the corresponding optimized witness is at least as powerful as the
 original one.  The same conclusion holds for the witnesses in
 Inequalities~\eqref{E:optimizing_non-Hermitian_witness},
 \eqref{E:optimizing_non-positive_witness_for_unit-trace_Hermitian},
 and~\eqref{E:optimizing_entanglement_witness} through similar arguments.  On
 the other hand, this argument does not apply to the remaining witnesses
 reported in this Section.
 \label{Rem:powerful_witness}
\end{Rem}

\subsection{Performance Analysis}
\label{Subsec:witness_performance}
 Surely, it is impossible to check every single criterion in the witnesses
 stated in \Cref{Subsec:witness_statements} as there are infinitely many $m_1,
 m_2, n_i$'s and $g(x)$ involved.   Besides, no witness involving a finite
 number of trace powers is sufficient.  For example, if the eigenvalues of a
 non-Hermitian matrix $M$ are $\lambda_j = e^{2\pi i j/p}$ for $j = 1,2,\dots,
 p$, then $\Tr M^q = 0$ for $q = 1,2,\dots,p-1$.  In other words, information
 of $\Tr M^q$ with $q = 1,\dots,p-1$ is consistent with the zero matrix.
 Hence, no \NHW involving only these traces can identify $M$ as non-Hermitian.
 Similarly, there exists a non-positive Hermitian matrix $H$ with $\Tr H^q =
 0$ for $q = 1,2,\dots,p$ if $\dim H > p$.  Thus, no \NPW or \UTNPW involving
 only these traces can detect its non-positiveness.  All these are consistent
 with the well-known fact that determining the separability of a state is
 NP-hard~\cite{entanglement_witness}.

\subsubsection{Non-Hermitian Witness}
\label{Subsubsec:non-Hermitian_witness_performance}

 Here I will only consider a subset of \NHW{s} stated in
 \Cref{Thrm:non-Hermitian_witness} with $\Tr M^j$ with $j \le d_{\max} \equiv
 8$.  The reason for this choice will be apparent once the performance results
 are given.  More precisely, the following algorithm is used.

 \par\bigskip\noindent
 [Non-Hermitian Witness (\NHW) Algorithm]
 \begin{enumerate}
  \item Input a linear operator or a square matrix $M$ and the maximum degree
   $d_{\max}$ of the trace power to be used in \Cref{Thrm:Hermitian_eigens}.
   Set $d = 2$.
  \item Output that $M$ is non-Hermitian if any one of the following
   inequalities holds:
   \label{Alg_NH:iterate}
   \begin{enumerate}
    \item Inequalities~\eqref{E:imag_part_non-Hermitian_witness},
     \eqref{E:12_non-Hermitian_witness}
     or~\eqref{E:123_optimized_non-Hermitian_witness} where $d = q$, $d = 2n$
     or $d = 3m$ with $n$ being an even integer and $m$ being an odd integer;
    \item Inequality~\eqref{E:optimizing_non-Hermitian_witness_special_cases}
     with $d = 2q$ and $q$ is an integer, or with $d = \ell > m > n$ where
     $\ell$ and $n$ are even positive integers and $m$ is odd; or
    \item Inequality~\eqref{E:optimizing_non-Hermitian_witness} for $d =
     \sum_i n_i$ with $g(x) = 1$ and $n_i = 2$ for all $i$.  (To speed things
     up, this should be the last condition to test.  Moreover one should stop
     once a set of $c_i$'s that obeying
     Inequality~\eqref{E:optimizing_non-Hermitian_witness} is found for such a
     non-optimal solution is already enough to guarantee non-Hermitianness.  I
     addition, for the special case of $d = 2$, one may simplify
     Inequality~\eqref{E:optimizing_non-Hermitian_witness} to
     Inequality~\eqref{E:reduced_optimizing_non-Hermitian_witness}.)
     \label{Alg_NH:generic}
   \end{enumerate}
  \item If $d > d_{\max}$, then output that the nature of $M$ is inconclusive.
   Otherwise, increase $d$ by $1$ and repeat \Cref{Alg_NH:iterate}.
 \end{enumerate}

 In simple terms, the \NHW Algorithm above checks if $M$ is non-Hermitian by
 gradually increasing the maximum power of its trace used in the witnesses.
 Furthermore, for each trace power, it systematically tests its
 non-Hermitianness through \Cref{Thrm:non-Hermitian_witness} by first checking
 $M$ for those easy conditions that require no optimization before testing the
 validity of computationally demanding
 Inequality~\eqref{E:optimizing_non-Hermitian_witness} with $g(x) = 1$ and
 $n_i = 2$ through optimization over $c_i$'s.

 Now I study the performance of this Algorithm.  Suppose the eigenvalues of an
 ensemble of matrices are uniformly distributed over the complex plane, then
 their traces are almost surely not real.  Thus, \Cref{Alg_NH:iterate} in this
 Algorithm can almost certainly certify that they are not Hermitian.  This is
 not surprising.

 I further demonstrate the power of this Algorithm on those non-Hermitian
 matrices whose eigenvalues come in conjugate pairs.  Showing the
 non-Hermitianness for this type of matrices is much harder for all the trace
 powers are real.  In particular, it is natural to study the performance the
 \NHW Algorithm on non-Hermitian matrices whose complex conjugate pairs of
 eigenvalues are uniformly distributed on the entire complex plane.

\begin{table*}[th]
 \centering
 \begin{tabular}{c|c|rrrrrrrr}
   \toprule
   \multirow{2}{*}{$d$} & \multirow{2}{*}{witness}
    & \multicolumn{8}{c}{$\dim M$}
  \\
   \cline{3-10}
   & & 2 & 4 & 6 & 8 & 10 & 20 & 50 & 100
  \\
   \midrule
   2 & Inequality~\eqref{E:012_optimizing_non-Hermitian_witness_special_cases}
    & $100.0\%$ & $73.0\%$ & $66.8\%$ & $63.6\%$ & $61.8\%$ & $57.9\%$
    & $54.9\%$ & $53.5\%$
  \\
   \midrule
   3 & Inequality~\eqref{E:123_optimized_non-Hermitian_witness} & & $1.9\%$
    & $2.2\%$ & $2.3\%$ & $2.4\%$ & $2.5\%$ & $2.7\%$ & $2.7\%$
  \\
   \midrule
   \multirow{5}{*}{4}
   & Inequality~\eqref{E:2_optimizing_non-Hermitian_witness_special_cases}
   & & $10.8\%$ & $13.8\%$ & $15.1\%$ & $15.8\%$ & $17.7\%$ & $19.5\%$
   & $20.5\%$
  \\
   & Inequality~\eqref{E:12_non-Hermitian_witness} & & $1.5\%$ & $3.1\%$
    & $3.9\%$ & $4.3\%$ & $4.7\%$ & $4.4\%$ & $4.1\%$
  \\
   & Inequality~\eqref{E:012_optimizing_non-Hermitian_witness_special_cases}
    & & $6.2\%$ & $7.3\%$ & $7.6\%$ & $7.5\%$ & $6.8\%$ & $5.6\%$ & $4.9\%$
  \\
   & $[4,3,2]$ & & $3.6\%$ & $2.7\%$ & $2.5\%$ & $2.3\%$ & $2.9\%$
    & $4.0\%$ & $4.7\%$
  \\
   & generic & & $3.0\%$ & $1.6\%$ & $1.3\%$ & $1.2\%$ & $1.2\%$ & $1.1\%$
    & $0.9\%$
  \\
   \midrule
   \multirow{6}{*}{6}
   & Inequality~\eqref{E:2_optimizing_non-Hermitian_witness_special_cases} & &
    & $0.3\%$ & $0.7\%$ & $1.1\%$ & $2.3\%$ & $3.4\%$ & $4.0\%$
  \\
   & Inequality~\eqref{E:012_optimizing_non-Hermitian_witness_special_cases} &
    & & $0.0\%$ & $0.1\%$ & $0.1\%$ & $0.1\%$ & $0.1\%$ & $0.1\%$
  \\
   & $[6,3,2]$ & & & $0.1\%$ & $0.1\%$ & $0.2\%$ & $0.2\%$ & $0.3\%$ & $0.4\%$
  \\
   & $[6,5,2]$ & & & $0.6\%$ & $0.9\%$ & $1.1\%$ & $1.5\%$ & $1.7\%$ & $1.8\%$
  \\
   & $[6,5,4]$ & & & $0.2\%$ & $0.3\%$ & $0.2\%$ & $0.2\%$ & $0.2\%$ & $0.2\%$
  \\
   & generic & & & $1.3\%$ & $1.5\%$ & $1.7\%$ & $1.6\%$ & $1.6\%$ & $1.6\%$
  \\
   \midrule
   \multirow{10}{*}{8}
   & Inequality~\eqref{E:2_optimizing_non-Hermitian_witness_special_cases} &
    & & & $0.0\%$ & $0.0\%$ & $0.1\%$ & $0.2\%$ & $0.2\%$
  \\
   & Inequality~\eqref{E:123_optimized_non-Hermitian_witness} & & & &
    $0.0\%$ & $0.0\%$ & $0.0\%$ & $0.0\%$ & $0.0\%$
  \\
   & Inequality~\eqref{E:012_optimizing_non-Hermitian_witness_special_cases} &
    & & & $0.0\%$ & $0.0\%$ & $0.1\%$ & $0.1\%$ & $0.1\%$
  \\ & $[8,3,2]$ & & & & $0.0\%$ & $0.0\%$ & $0.0\%$ & $0.0\%$ & $0.0\%$
  \\ & $[8,5,2]$ & & & & $0.0\%$ & $0.0\%$ & $0.0\%$ & $0.0\%$ & $0.0\%$
  \\ & $[8,5,4]$ & & & & $0.0\%$ & $0.0\%$ & $0.0\%$ & $0.0\%$ & $0.0\%$
  \\ & $[8,7,2]$ & & & & $0.0\%$ & $0.0\%$ & $0.1\%$ & $0.1\%$ & $0.1\%$
  \\ & $[8,7,4]$ & & & & $0.0\%$ & $0.0\%$ & $0.0\%$ & $0.0\%$ & $0.0\%$
  \\ & $[8,7,6]$ & & & & $0.0\%$ & $0.0\%$ & $0.0\%$ & $0.0\%$ & $0.0\%$
  \\ & generic & & & & $0.1\%$ & $0.1\%$ & $0.1\%$ & $0.1\%$ & $0.1\%$
  \\
   \midrule
   \multicolumn{2}{c|}{\ding{55}} & & & & & $0.0\%$ & $0.0\%$ & $0.0\%$
    & $0.0\%$
  \\
   \bottomrule
 \end{tabular}
 \caption{Performance of the \NHW Algorithm.  The eigenvalue distribution used
  are stated in the text.  Statistics are gathered using $10^6$ randomly
  generated samples for each Hilbert space dimension of $M$.  In the second
  column, ``generic'' means using \Cref{Alg_NH:generic} in the \NHW Algorithm,
  whereas $[\ell,m,n]$ refers to using
  Inequality~\eqref{E:lmn_optimizing_non-Hermitian_witness_special_cases} with
  parameters $\ell,m$ and $n$.  The witnesses here are listed in the actual
  order of the test used in my simulation.  Entries in each column are the
  break down of the successful probabilities of determining this sample of $M$
  to be non-Hermitian.  (Since the \NHW Algorithm returns a certificate once
  a valid \NHW is found, the successful probabilities listed here for a \NHW
  that is tested later may not truly reflect its power.  For instance,
  according to \Cref{Rem:powerful_witness} and data listed in this Table, the
  successful probability for using the generic \NHW with $d = 6$ for $\dim M =
  6$ is at least $1.6\% + 1.3\% = 2.9\%$.  This remark holds also for all
  other Tables in this paper.)  Error bar in each entry is about $0.1\%$.
  Here an entry of $0.0\%$ means that the observed successful probability is
  smaller than $0.1\%$ but non-zero; whereas an empty entry means that no
  successful case is observed.  Lastly, the column under \ding{55} tabulates
  the unsuccessful probabilities, namely those that fail to be certified when
  $d_{\max} = 8$.  Note that the column sum may deviate slightly from $100\%$
  due to rounding.
  \label{T:non-Hermitian_witness}
 }
\end{table*}

 The following observation is useful to generate these type of matrices.
 Suppose all $c_i$'s are scaled by a real factor $s > 0$.  Further suppose
 that $M$ is replaced by $s M$.  Then, from the derivation of
 \Cref{Thrm:Hermitian_eigens}, I know that the LHS of
 Inequality~\eqref{E:Hermitian_eigens} will be scaled by the factor $s^{\deg
 g(x) + \sum_i n_i}$.  As a consequence, all \NHW{s} in
 \Cref{Thrm:non-Hermitian_witness} are scale invariant in the sense that if
 they can certify that $M$ is non-Hermitian, then they can do so for $s M$.
 Surely, a natural choice for $s$ is the sum of magnitude square of all its
 eigenvalues.  Consequently, I could study the performance of the \NHW
 Algorithm above on non-Hermitian matrices whose complex conjugate pairs of
 eigenvalues are uniformly distributed on $\numset{C}$ by normalizing each of
 them by the sum of magnitude square of all their corresponding eigenvalues
 first.

 The results reported in \Cref{T:non-Hermitian_witness} show the effectiveness
 of this \NHW Algorithm.  (The Mathematica code used to produce all data in
 this paper is available on request.)  In fact, by fixing $d_{\max} = 8$, the
 observed unsuccessful probability gradually increases with $\dim M$ and is
 equal to about $0.006\%$ (which is non-zero though the uncertainty is huge)
 when $\dim M = 100$.  For these failure cases, the LHS of
 Inequality~\eqref{E:optimizing_non-Hermitian_witness} decreases from about
 $10^{-4}$ for $\dim M = 10$ to about $3\times 10^{-7}$ when $\dim M = 100$.
 That is, all these failure cases just slightly miss the bar of being a valid
 certificate their non-Hermitianness.  Moreover, essentially all the fail
 cases in the simulation are the ones whose arguments of their dominant
 eigenvalues (in the sense that their magnitudes to the $d_{\max}$th power are
 much larger than the rest) are close to either $0$ or $\pm\pi$.  This can be
 understood as follows.  The contribution of non-dominant eigenvalues to $\Tr
 M^j$ diminishes as $j$ increases.  In other words, $\Tr M^j$ is closer and
 closer to that of the Hermitian matrix whose eigenvalues equal the real part
 of the dominant eigenvalues of $M$ as $j$ increases until $j$ is large
 enough for the imaginary parts of the dominant eigenvalues of $M$ to start
 contributing non-negligibly to $\Tr M^j$.  Let me denote this large enough
 $j$ by $j_c$.  Suppose the eigenvalues of $M$ are distributed in such a way
 that all the \NHW{s} used in the above Algorithm involving traces up to $d$th
 power fail for $d < j_c$.  Then, \NHW{s} involving traces up to $d'$th power
 obeying $d < d' < j_c$ should fail, too.  Obviously, the closer the dominant
 eigenvalues of $M$ are to the real line, the higher the value of $j_c$ is.
 Two comments are in order.  First, a much rarer failure case is when the
 arguments of the dominant eigenvalues are approximately equally spaced with
 spacing interval $\approx 2\pi/q$ where $q$ is an integer greater than
 $d_{\max}$.  (This is a straightforward extension of the construction
 reported earlier in \Cref{Subsec:witness_performance}.  In this case, all
 \NHW{s} involving trace up to $j_c$th power of $M$ fail for they are unable
 to distinguish this $M$ with the zero matrix.)  Since the eigenvalues are
 uniformly distributed in $\numset{C}$ before normalization, the occurrence
 probability is extremely small.  Second, by the same token, one of the
 hardest case to detect is when the dominant eigenvalues are all real although
 its occurrence probability is almost surely zero.

 Two remarks are in place.  First, because the \NHW Algorithm stops once a
 valid non-Hermitian certification is found, the successful probability
 distribution in \Cref{T:non-Hermitian_witness} has to be interpreted
 carefully.  To be precise, probabilities listed in later rows of the Table
 are conditional probabilities given that all tests listed in the earlier part
 of the Table fail.  That is to say, a small probability shown in the lower
 part of the Table need not imply the ineffectiveness of the corresponding
 witness.  This remark applies to all subsequent Tables in this paper.
 Second, for a fixed $\dim M$, the non-Hermitianness of up to at least
 $97.0\%$ of the cases in \Cref{T:non-Hermitian_witness} can be certified
 using witnesses without requiring explicit optimization of $c_i$'s.  In other
 words, the average run time is extremely fast.

 The absence of entries with odd $d \ge 5$ in \Cref{T:non-Hermitian_witness}
 means that those witnesses are ineffective for the distribution of
 non-Hermitian matrices described earlier in this Subsubsection.  From the
 \NHW Algorithm, these \NHW{s} come from
 Inequality~\eqref{E:123_optimized_non-Hermitian_witness} with $m > 1$.
 \Cref{T:non-Hermitian_witness} also reveals that for a fixed $d_{\max}$, the
 successful probability for using $\Tr M^j$ with $j \le d_{\max}$ gradually
 decreases as $\dim M$ increases.  The reason is that $\Tr M^j \to 0$ for any
 fixed $j\in \numset{Z}^+$ as $\dim M \to \infty$ since eigenvalues are
 uniformly distributed on $\numset{C}$.  Furthermore, as $\dim M$ increases,
 the \NHW{s} corresponding to $d=2$ become less effective while those
 corresponding to $d>2$ are generally getting more effective.  Besides, the
 larger the dimension of $M$, the greater the value of $d_{\max}$ is required
 to attain the same successful probability.  These observations are within
 expectation.  \Cref{T:non-Hermitian_witness} also shows that
 Inequality~\eqref{E:012_optimizing_non-Hermitian_witness_special_cases} is
 the most powerful \NHW for $d = 2$; and
 Inequality~\eqref{E:2_optimizing_non-Hermitian_witness_special_cases} is the
 most effective one for $d = 4$.  These two Inequalities combined are already
 able to certify at least $74.0\%$ of the cases up to $\dim M = 100$.  As for
 cases that require $d \ge 6$, the optimized witness in \Cref{Alg_NH:generic}
 works better than the rest.  Other than these, there is no obvious trend or
 pattern for the effectiveness of the \NHW{s} used.

\subsubsection{Non-Positive Witness}
\label{Subsubsec:non-positive_witness_performance}
 Likewise, it is instructive to study the performance of \NPW and \UTNPW over
 a ``uniformly distributed'' sample of Hermitian matrices.  Although there is
 no natural way to define a uniformly distribution in this context, I could
 proceed as follows.  Using a similar argument in
 \Cref{Subsubsec:non-Hermitian_witness_performance}, I know that the
 \UTNPW{s} in \Cref{Thrm:non-positive_witness_for_unit-trace_Hermitian}
 together with the \NPW{s} in
 \Cref{Thrm:non-positive_witness_for_general_Hermitian} are scale invariant.
 Combined with the fact that these witnesses depend only on the eigenvalues
 rather than eigenvectors of $H$, I may define  ``uniformly distributed''
 Hermitian matrices in this context by normalizing \iid eigenvalues with an
 overall normalization factor.  More precisely, I use the following method to
 generate random non-positive Hermitian matrices.  First, $D \equiv \dim H$
 statistically independent uniformly distributed real random numbers
 $\lambda_j$'s in $[-1,1]$ are drawn.  These $D$ numbers are then normalized
 in the $\ell_2$-norm.  That is, I demand that $\sum_j \lambda_j^2 = 1$
 through scaling.  The resultant $\lambda_j$'s are the eigenvalues of the
 randomly generated Hermitian matrix $H$.  Finally, only those with at least
 one negative eigenvalue are kept.

 \par\bigskip\noindent
 [Non-Positivity Witness (\UTNPW) Algorithm For Unit Trace Hermitian Matrix]
 \begin{enumerate}
  \item Input a unit trace Hermitian operator or matrix $H$ and the maximum
   degree $d_{\max}$ for the polynomial $G_{NN}$ to be used in
   \Cref{Thrm:positive_eigens_for_unit_trace_H}.  Set $d = 2$ and $m_1 = 1$.
  \item Output that $H$ is non-positive if any one of the inequalities below
   is true:
   \label{Alg_NP:iterate}
   \begin{enumerate}
    \item Inequality~\eqref{E:12_non-positive_witness_for_unit-trace_Hermitian}
     where $d = 2m$ with $m$ being an odd integer;
   non-positive.
    \item Inequality~\eqref{E:optimizing_non-positive_witness_for_general_Hermitian_special_cases}
     with $d = m$ and $m$ is an odd integer, or with $d = 2(q+1)$ and $q \in
     \numset{Z}^+$, or with $d = m > n > t$ where $m$ and $t$ are odd positive
     integers and $n$ is even; or
    \item Inequality~\eqref{E:optimizing_non-positive_witness_for_unit-trace_Hermitian}
     for $d = m_1 + m_2 + \sum_i n_i$ with $g(x) = 1$ and $n_i = 2$ for all $i$
     where $m_2 = 0$ or $1$.  (To speed things up, this should be the last
     condition to test.  Moreover, one could stop once a set of $c_i$'s that
     obeying
     Inequality~\eqref{E:optimizing_non-positive_witness_for_unit-trace_Hermitian}
     is found rather than finding the optimal set of $c_i$'s.  In addition,
     for the special case of $d = 4$ and $m_2 = 1$, one may simplify
     Inequality~\eqref{E:optimizing_non-positive_witness_for_unit-trace_Hermitian}
     to
     Inequality~\eqref{E:reduced_optimizing_non-positive_witness_for_unit-trace_Hermitian}.)
     \label{Alg_NP:generic}
   \end{enumerate}
  \item If $d > d_{\max}$, then output that the nature of $H$ is inconclusive.
   Otherwise, increase $d$ by $1$ and repeat \Cref{Alg_NP:iterate}.
 \end{enumerate}

\begin{table*}[th]
 \centering
 \begin{tabular}{c|c|rrrrrrrrr}
   \toprule
   \multirow{2}{*}{$d$} & \multirow{2}{*}{witness}
    & \multicolumn{9}{c}{$\dim H$}
  \\
   \cline{3-11}
   & & 2 & 3 & 4 & 5 & 6 & 7 & 10 & 50 & 100
  \\
   \midrule
   2 & $[2,1,0]_1$ & $100.0\%$ & $79.5\%$ & $72.2\%$ & $68.8\%$ & $67.4\%$
    & $66.8\%$ & $66.7\%$ & $68.0\%$ & $68.1\%$
  \\
   \midrule
   \multirow{2}{*}{3}
    & Inequality~\eqref{E:m_optimizing_non-positive_witness_for_general_Hermitian_special_cases}
    & & & & & $0.0\%$ & $0.1\%$ & $0.1\%$ & $0.1\%$ & $0.1\%$
  \\
   & $[3,2,1]_0$ & & $16.3\%$ & $23.0\%$ & $26.9\%$ & $29.3\%$ & $30.6\%$
    & $32.3\%$ & $31.9\%$ & $31.8\%$
  \\
   \midrule
   4 & $\text{generic}_1$ & & $1.4\%$ & $1.3\%$ & $1.0\%$ & $0.7\%$ & $0.5\%$
    & $0.1\%$ & &
  \\
   \midrule
   \multirow{2}{*}{5} & $[5,4,3]_0$ & & & $0.0\%$ & $0.0\%$ & $0.0\%$
    & $0.0\%$ & $0.0\%$ & &
  \\
   & $\text{generic}_0$ & & $2.8\%$ & $2.9\%$ & $2.6\%$ & $2.1\%$ & $1.5\%$
    & $0.6\%$ & &
  \\
   \midrule
   \multirow{2}{*}{6}
    & Inequality~\eqref{E:12_non-positive_witness_for_unit-trace_Hermitian} &
    & & $0.0\%$ & $0.0\%$ & $0.0\%$ & & & &
  \\
   & $\text{generic}_1$ & & & $0.1\%$ & $0.1\%$ & $0.1\%$ & $0.1\%$ & $0.0\%$
    & &
  \\
   \midrule
   7 & $\text{generic}_0$ & & & $0.6\%$ & $0.5\%$ & $0.4\%$ & $0.3\%$
    & $0.1\%$ & &
  \\
   \midrule
   8 & $\text{generic}_1$ & & & & $0.0\%$ & $0.0\%$ & $0.0\%$ & $0.0\%$ & &
  \\
   \midrule
   \multicolumn{2}{c|}{\ding{55}} & & & & $0.1\%$ & $0.2\%$ & $0.1\%$
    & $0.1\%$ &
  \\
   \bottomrule
 \end{tabular}
 \caption{Performance of the \UTNPW Algorithm on unit trace non-positive
  Hermitian matrices.  The eigenvalue distribution used is stated in the text.
  In the second column, ``$\text{generic}_{m_2}$'' means using
  \Cref{Alg_NP:generic} in the \UTNPW Algorithm; whereas $[m,n,t]_{m_2}$
  refers to using
  Inequality~\eqref{E:mnt_optimizing_non-positive_witness_for_general_Hermitian_special_cases}
  with parameters $m,n$ and $t$.  In both cases, subscript $m_2$ refers to the
  exponent of the $(1-x)$ factor used in $G_{NN}(x)$.  Other parameters and
  conventions used here are the same as those in
  \Cref{T:non-Hermitian_witness}.
  \label{T:unit-trace_non-positive_witness}
 }
\end{table*}

\begin{table*}[th]
 \centering
 \begin{tabular}{c|c|rrrrrrrrr}
   \toprule
   \multirow{2}{*}{$d$} & \multirow{2}{*}{witness}
    & \multicolumn{9}{c}{$\dim H$}
  \\
   \cline{3-11}
   & & 2 & 3 & 4 & 5 & 6 & 7 & 10 & 50 & 100
  \\
   \midrule
   1
    & Inequality~\eqref{E:m_optimizing_non-positive_witness_for_general_Hermitian_special_cases}
    & $50.0\%$ & $50.0\%$ & $50.0\%$ & $50.0\%$ & $50.0\%$ & $50.0\%$
    & $50.0\%$ & $50.0\%$ & $50.0\%$
  \\
   \midrule
   2
    & Inequality~\eqref{E:12_optimizing_non-positive_witness_for_general_Hermitian}
    & $50.0\%$ & $39.8\%$ & $36.2\%$ & $34.4\%$ & $33.7\%$ & $33.4\%$
    & $33.3\%$ & $33.9\%$ & $34.1\%$
  \\
   \midrule
   3 & $[3,2,1]$ & & $8.1\%$ & $11.5\%$ & $13.5\%$ & $14.6\%$ & $15.3\%$
    & $16.3\%$ & $16.0\%$ & $15.9\%$
  \\
   \midrule
   \multirow{2}{*}{5} & $[5,4,3]$ & & & $0.0\%$ & $0.0\%$ & $0.0\%$ & $0.0\%$
    & $0.0\%$ & &
  \\
   & generic & & $2.1\%$ & $2.1\%$ & $1.7\%$ & $1.4\%$ & $1.0\%$ & $0.4\%$ & &
  \\
   \midrule
   7 & generic & & & $0.3\%$ & $0.3\%$ & $0.2\%$ & $0.2\%$ & $0.1\%$ & &
  \\
   \midrule
   \multicolumn{2}{c|}{\ding{55}} & & & & $0.1\%$ & $0.1\%$ & $0.1\%$
    & $0.0\%$ &
  \\
   \bottomrule
 \end{tabular}
 \caption{Performance of the \NPW Algorithm on general non-positive Hermitian
  matrices.  The eigenvalue distribution used is stated in the text.
  In the second column, ``generic'' means using \Cref{Alg_NPG:generic} in the
  \NPW Algorithm; whereas $[m,n,t]$ refers to using
  Inequality~\eqref{E:mnt_optimizing_non-positive_witness_for_general_Hermitian_special_cases}
  with parameters $m,n$ and $t$.  Other parameters and conventions used here
  are the same as those in \Cref{T:non-Hermitian_witness}.
  \label{T:general_non-positive_witness}
 }
\end{table*}

 \par\bigskip\noindent
 [Non-Positivity Witness (\NPW) Algorithm For General Hermitian Matrix]
 \begin{enumerate}
  \item Input a Hermitian operator or matrix $H$ and the maximum degree
   $d_{\max}$ for the polynomial $G'_{NN}$ to be used in
   \Cref{Cor:positive_eigens}.  Set $d = 1$ and $m_1 = 1$.
  \item Output that $H$ is non-positive if any one of the inequalities below
   is true:
   \label{Alg_NPG:iterate}
   \begin{enumerate}
    \item Inequality~\eqref{E:12_optimizing_non-positive_witness_for_general_Hermitian}
     where $d = 2m$ with $m$ being an odd integer;
    \item Inequality~\eqref{E:optimizing_non-positive_witness_for_general_Hermitian_special_cases}
     with $d = m$ and $m$ is an odd integer, or with $d = 2(q+1)$ and $q \in
     \numset{Z}^+$, or with $d = m > n > t$ where $m$ and $t$ are odd positive
     integers and $n$ is even; or
    \item Inequality~\eqref{E:optimizing_non-positive_witness_for_general_Hermitian}
     for $d = m_1 + \sum_i n_i$ with $g(x) = 1$ and $n_i = 2$ for all $i$.
     (To speed things up, this should be the last condition to test.
     Moreover, one could stop once a set of $c_i$'s that obeying
     Inequality~\eqref{E:optimizing_non-positive_witness_for_general_Hermitian}
     is found rather than finding the optimal set of $c_i$'s.)
     \label{Alg_NPG:generic}
   \end{enumerate}
  \item If $d > d_{\max}$, then output that the nature of $H$ is inconclusive.
   Otherwise, increase $d$ by $1$ and repeat \Cref{Alg_NPG:iterate}.
 \end{enumerate}

 The successful detection probabilities of the \NPW{s} for unit trace and
 general non-positive Hermitian matrices of various dimensions generated by
 the above methods are listed in
 \Cref{T:unit-trace_non-positive_witness,T:general_non-positive_witness},
 respectively.  As shown in both Tables, using $d_{\max} = 3$ is enough to
 certify all cases in the $10^6$ sample non-positive Hermitian matrices when
 $\dim H \gg 10$.  Besides, the corresponding \UTNPW{s} or \NPW{s} used do not
 need to be explicitly optimized over $c_j$'s and hence are extremely
 efficient to compute.  Note that for both the \UTNPW and \NPW, the successful
 probabilities of witnesses corresponding to $d \ge 3$ appear to saturate as
 $\dim H \to \infty$.  In addition,
 \Cref{T:unit-trace_non-positive_witness,T:general_non-positive_witness} show
 that the hardest cases to handle are the ones with $\dim H$ between $5$ to
 about $10$.  In these cases, using $d$ up to $3$ is sufficient to correctly
 certify non-positivity $96.5\%$ of the time.  Besides, by using up to trace
 of the 7th or the 8th power of the matrix, the unsuccessful non-positiveness
 detection rates are about $0.2\%$ at most.  For these failure cases of the
 \UTNPW Algorithm, the LHS of
 Inequality~\eqref{E:optimizing_non-positive_witness_for_unit-trace_Hermitian}
 is less than or equal to about $8\times 10^{-7}$; whereas for those of the
 \NPW Algorithm, the LHS of
 Inequality~\eqref{E:optimizing_non-positive_witness_for_general_Hermitian} is
 less than or equal to about $10^{-4}$.  In other words, they are just
 slightly off from being valid witnesses.  Actually, almost all failure cases
 contains all but one positive eigenvalues.  (Very occasionally, two negative
 eigenvalues are found in the failure cases.)  Besides, the magnitude of their
 only negative eigenvalue is very small in comparison to those of the positive
 eigenvalues.  The reason is simple --- $\Tr H^j$ is approximately equal to
 that of $\Tr \bar{H}^j$ for all $j \in \numset{Z}^+$ where $\bar{H}$ denotes
 the non-negative Hermitian matrix with all the negative eigenvalues of $H$
 replaced by $0$.  Their closeness in trace powers make it difficult to
 certify the non-positiveness of $H$.

 To understand the structure of the fail cases, from the proof technique of
 \Cref{Thrm:Hermitian_eigens}, the LHS of
 Inequality~\eqref{E:non-positive_witness_for_unit-trace_Hermitian} equals
 $\sum_j G_{NN}(\lambda_j)$ before optimization over $c_i$'s where
 $\lambda_j$'s are the eigenvalues of $H$.  By the construction of $G_{NN}(x)$
 in \Cref{E:G_NN_def}, $G_{NN}(x) < 0$ if $x < 0$.  Therefore, the hardest
 case to detect non-positivity using
 Inequality~\eqref{E:non-positive_witness_for_unit-trace_Hermitian} is the
 one having only a single negative eigenvalue, say, $\lambda_1$.  Consider the
 polynomial $G_{NN}(x)$.  Since Runge's phenomenon generally shows up in a
 polynomial with degree $\gtrsim 4$, one expects that for a general set of
 $c_i$'s (and hence also the optimized set), $G_{NN}(x)$ fluctuates widely
 around $0$ when $\sum_i n_i \gtrsim 4$.  Besides, the smaller the value of
 $m_1$ or the larger the values of $n_i$'s, the larger the fluctuation.
 Provided that this fluctuation is sufficiently negative around $x = 0$, one
 obtains a valid \UTNPW.  In other words, Runge's phenomenon is exploited to
 construct a powerful witness.  Clearly, the amplitude of the fluctuation is
 large when $m_1$ is small and $\lambda_1$ is close to $0$.  That is why
 Inequality~\eqref{E:non-positive_witness_for_unit-trace_Hermitian} is a
 powerful \UTNPW when $m_1 = 1$ and $d \gtrsim 5$.  Furthermore, it explains
 why the hardest case to deal with is the one with only one negative
 eigenvalue whose magnitude is small compared to those of the positive
 eigenvalues.  Using the same argument, I expect that the \NHW{s}, \NPW{s} and
 \EW{s} reported here are effective when $d \gtrsim 4$.  These claims are
 confirmed by simulation results reported in this Section.  In particular, by
 setting $d_{\max} = 8$, the successful probabilities are at least $99.8\%$ in
 all cases that I have tested.

\begin{table}[th]
 \centering
 \begin{tabular}{c|c|rrrrrrr}
   \toprule
   \multirow{2}{*}{$d$} & \multirow{2}{*}{witness}
    & \multicolumn{7}{c}{$D \times D$}
  \\
   \cline{3-9}
   & & $2\times 2$ & $3\times 3$ & $4\times 4$ & $5\times 5$ & $6\times 6$
    & $7\times 7$ & $8\times 8$
  \\
   \midrule
   3 & $[3,2,1]_0$ & $29.1\%$ & $25.3\%$ & $23.2\%$ & $21.9\%$ & $21.0\%$
    & $20.4\%$ & $19.9\%$
  \\
   \midrule
   4 & $\text{generic}_1$ & $13.8\%$ & $14.4\%$ & $14.8\%$ & $14.8\%$
    & $14.7\%$ & $14.6\%$ & $14.5\%$
  \\
   \midrule
   5 & $\text{generic}_0$ & $43.0\%$ & $58.0\%$ & $62.0\%$ & $63.3\%$
    & $64.3\%$ & $64.9\%$ & $65.6\%$
  \\
   \midrule
   6 & $\text{generic}_1$ & $1.9\%$ & $0.6\%$ & $0.0\%$ & & & &
  \\
   \midrule
   7 & $\text{generic}_0$ & $12.2\%$ & $1.8\%$ & $0.0\%$ & & & &
  \\
   \midrule
   8 & $\text{generic}_1$ & & $0.0\%$ & & & & &
  \\
   \midrule
   \multicolumn{2}{c|}{\ding{55}} & & $0.2\%$ & & & & &
  \\
   \bottomrule
 \end{tabular}
 \caption{Performance of the \EW Algorithm on density matrices that whose
  \PPT{s} are non-positive.  The generation method of these density matrices
  is stated in the text.  In the second column,
  ``$\text{generic}_{m_2}$'' means using
  Inequality~\eqref{E:optimizing_entanglement_witness}, whereas
  $[m,n,t]_{m_2}$ refers to the values of $m,n$ and $t$ used in
  Inequality~\eqref{E:mnt_optimizing_non-positive_witness_for_general_Hermitian_special_cases}.
  In both cases, subscript $m_2$ refers to the exponent of the $(1-x)$ factor
  used in $G_{NN}(x)$.  The parameters and other conventions used are the same
  as those in \Cref{T:non-Hermitian_witness}.
  \label{T:entanglement_witness}
 }
\end{table}

 \Cref{T:unit-trace_non-positive_witness} shows that for the \UTNPW{s} derived
 from $G_{NN}(x)$ with $m_1 = 1$, the value of $m_2$ must equal $(d+1) \bmod
 2$.  This is a simple consequence of the degree counting of the polynomial
 $G_{NN}(x)$ used.  Interestingly, \UTNPW{s} corresponding to $m_1 = 0$ and
 $m_2 = 1$ are ineffective.  Precisely, it means that whenever a \UTNPW
 corresponding to $m_1 = 0$ and $m_2 = 1$ can certify the non-positivity of a
 Hermitian matrix, so does the \UTNPW corresponding to $m_1 = 1$ and $m_2 =
 0$.  Unfortunately, I do not have a good explanation.  From
 \Cref{T:unit-trace_non-positive_witness,T:general_non-positive_witness},
 Inequality~\eqref{E:mnt_optimizing_non-positive_witness_for_general_Hermitian_special_cases}
 is a useful witness when $d = 2,3,5$ and $d = 3,5$, respectively.  Moreover,
 Inequality~\eqref{E:1q2q_optimizing_non-positive_witness_for_general_Hermitian_special_cases}
 is not useful at all.  The successful probability of the \NPW using
 Inequality~\eqref{E:m_optimizing_non-positive_witness_for_general_Hermitian_special_cases}
 for $d = 1$ (that is, $\Tr H < 0$) equals $50\%$.  This is due to the fact
 that there is an equal chance to produce a non-positive Hermitian matrix of
 positive or negative trace using the construction described in this
 Subsubsection.  In contrast, as shown in
 \Cref{T:unit-trace_non-positive_witness,T:general_non-positive_witness},
 Inequality~\eqref{E:m_optimizing_non-positive_witness_for_general_Hermitian_special_cases}
 is no longer an effective witness when $d > 3$ in the \UTNPW Algorithm and
 $d > 1$ in the \NPW Algorithm, respectively.  Interestingly,
 \Cref{T:general_non-positive_witness} depicts that all even-numbered $d > 2$
 are not useful in obtaining better \NPW witnesses in the \NPW Algorithm for
 the non-positive Hermitian matrices generated using the method described
 earlier in this Subsubsection.  This means that
 Inequalities~\eqref{E:12_optimizing_non-positive_witness_for_general_Hermitian}
 and~\eqref{E:1q2q_optimizing_non-positive_witness_for_general_Hermitian_special_cases}
 are not powerful \NPW{s} for even-numbered $d \ge 4$.  In contrast,
 Inequality~\eqref{E:12_optimizing_non-positive_witness_for_general_Hermitian}
 is an extremely effective \NPW when $d = 2$.

\subsubsection{Entanglement Witness}
\label{Subsubsec:entanglement_witness_performance}
 The \UTNPW Algorithm is turned to the Entanglement Witness (\EW) Algorithm by
 using the \PPT of a density matrix as input.  To investigate the performance
 of this Algorithm, I first generate bipartite density matrices with uniformly
 distributed eigenvalues.  This can be done by generating a Ginibre ensemble
 in both the real and imaginary parts of each matrix element for a matrix $G$.
 In other words, every element $G_{jk}$ follows $\normaldist + i
 \,\normaldist$ where $\normaldist$ denotes the standard normal distribution.
 Moreover, these elements are statistically independent.  Then, $G G^\dag /
 \Tr ( G G^\dag )$ is the required density matrix
 $\rho$~\cite{Random_Matrices}.  The corresponding \PPT matrix $\sigma$ can be
 computed readily.  Since no non-positiveness test on $\sigma =
 \rho^\text{\PPT}$ can certify entanglement if $\sigma \ge 0$, it makes more
 sense here to gauge the performance of this \EW Algorithm in certifying
 entanglement for only those whose $\sigma$ is non-positive.  In fact, this
 was the performance evaluation approach adopted in, for example,
 Refs.~\cite{EKetal20,YIG21}.

 The results for $D\times D$ density matrices with non-positive \PPT{s} are
 tabulated in \Cref{T:entanglement_witness}.  It shows that for $D \ge 5$,
 using $d_{\max} = 5$ is already sufficient to the certify entanglement for
 those density matrices with non-positive \PPT{s}.  The hardest case is for $D
 = 3$ in which the successful probability is about $99.8\%$, which is better
 than $99.0\%$ success rate method reported in Ref.~\cite{YIG21}.  (For all
 the other cases investigated, these two methods are equally good.)  Just like
 the \UTNPW Algorithm situation, the failure cases are the ones with $\sigma$
 having exactly one negative eigenvalue.  Furthermore, the LHS of
 Inequality~\eqref{E:optimizing_entanglement_witness} is about $\lesssim 2
 \times 10^{-7}$ which is very close to be a valid certification.

 Interestingly, all useful \EW{s} originate
 from Inequality~\eqref{E:optimizing_entanglement_witness}.  Therefore,
 certifying entanglement through the \EW Algorithm is computationally
 demanding.  And just like the situation of \UTNPW, the \EW corresponding to
 $m_1 = 0$ and $m_2 = 1$ is ineffective.

 Finally, there is a word of caution.  The differences between
 \Cref{T:unit-trace_non-positive_witness} and \Cref{T:entanglement_witness} 
 imply that performance of various witness algorithms depends strongly on the
 eigenvalue distribution of the sample used.  Thus, one should never skip some
 of the witness tests in all the Algorithms mentioned in this Section if the
 eigenvalue distribution of the input is not known.

\subsection{Practical Numerical Considerations}
\label{Subsec:numerical_issue_performance}
 Using the same analysis and notations in
 \Cref{Subsec:trace_powers_numerics}, one can safely certify entanglement of a
 density matrix $\rho$ (whose \PPT is $\sigma$) through the witness in
 Inequality~\eqref{E:optimizing_entanglement_witness} if its LHS equals
 $-\epsilon$ with $\epsilon > 0$ obeying
\begin{equation}
 \sum_j |b_j| ( 2u |\Tr \sigma^j| + |\delta_j|) < \epsilon .
 \label{E:forward_error_EW}
\end{equation}
 Let me consider the most difficult cases that can be certified by
 \Cref{Subsubsec:entanglement_witness_performance}.  They are the ones for
 $\dim H = 3$, $d = 8$, $c_i$'s range from about $0.06$ to $0.30$, the
 smallest eigenvalue $\approx -10^{-2}$, and LHS of
 Inequality~\eqref{E:optimizing_entanglement_witness} ranges from $-10^{-9}$
 to $-2\times 10^{-8}$.  In most experiments, one expects that the measurement
 uncertainty $|\delta_j| \sim \delta |\Tr H^j|$.  Using typical values of
 $c_i$'s and $\lambda_j$'s obtained in simulation, I arrive at $\sum_j |b_j|
 (2u |\Tr \sigma^j| + |\delta_j|) = \sum_j |b_j| \,|\Tr \sigma^j| (2u +
 \delta) \approx \delta \sum_j |b_j| \,|\Tr \sigma^j| \approx 5\times
 10^{-3}$.  Hence, \Cref{E:forward_error_EW} demands that the certification is
 sound if $\delta \lesssim 200\epsilon = 2 \times 10^{-7}$.  This accuracy
 requires a lot of independent sample runs and is barely within reach to date.
 To be fair, this is not completely surprising as distinguishing an
 eigenvalue of order of $-10^{-2}$ from $0$ for a trace one matrix in an
 experiment is never easy.  Using the same argument, certifying those
 difficult cases for non-Hermitianness or non-positiveness reported here is
 hard but not undoable in practice.

\section{Conclusions}
\label{Sec:conclusion}
 In summary, I have conducted a systematic study of constraining eigenvalues
 of Hermitian matrices through trace of their powers plus the associated \NHW,
 \UTNPW, \NPW and \EW.  Their strengths and weaknesses have also been
 analyzed.  Hard to certify cases have been classified and witness
 certification in the presence of rounding error and measurement uncertainty
 have been discussed.  It is instructive to study the possibility of extending
 this method to cover other kind of classifications.  One possibility is the
 so-called imaginary witness~\cite{imag_witness}.

\begin{acknowledgments}
 Helpful discussions with Wenyuan (Mike) Wang is gratefully acknowledged.
\end{acknowledgments}

\bibliographystyle{apsrev4-2}

\bibliography{qc84.1}

\end{document}